\documentclass[runningheads,a4paper]{llncs}

\usepackage{amssymb,amsmath}

\usepackage{graphicx}
\usepackage{color}
\usepackage{colortbl}
\usepackage{bm}
\usepackage{complexity}
\usepackage{framed}
\usepackage{enumitem}
\usepackage[caption=false]{subfig}
\usepackage{float}
\usepackage{hyperref}
\usepackage{cite}
\usepackage{wasysym}
\usepackage[colorinlistoftodos,prependcaption,textsize=tiny]{todonotes}
\usepackage[latin1]{inputenc}

\newtheorem{innercustomgeneric}{\customgenericname}
\providecommand{\customgenericname}{}
\newcommand{\newcustomtheorem}[2]{%
  \newenvironment{#1}[1]
  {%
   \renewcommand\customgenericname{#2}%
   \renewcommand\theinnercustomgeneric{##1}%
   \innercustomgeneric
  }
  {\endinnercustomgeneric}
}

\newcustomtheorem{customprop}{Proposition}
\newcustomtheorem{customlmm}{Lemma}
\newcustomtheorem{customdef}{Definition}
\newcustomtheorem{customthm}{Theorem}

\newcounter{theonew} 
 
\newcounter{theonew3}

\newtheorem{theo}[theonew]{Theorem}

\newtheorem{constr}[theonew3]{Construction}

\usepackage[english,algoruled,longend,ruled,vlined,linesnumbered]{algorithm2e}

\def\qed{{\hfill\hbox{\rlap{$\sqcap$}$\sqcup$}}}

\usepackage{url}

\newcommand{\keywords}[1]{\par\addvspace\baselineskip
\noindent\keywordname\enspace\ignorespaces#1}
\begin{document}

\mainmatter  
\title{The graph tessellation cover number: extremal bounds, efficient algorithms and hardness\thanks{This work was partially supported by the Brazilian agencies CAPES, CNPq and FAPERJ.}}

\titlerunning{The graph tessellation cover number}

\author{A.~Abreu\inst{1}
\and L.~Cunha\inst{2}\and T.~Fernandes\inst{3, 4} \and C.~de Figueiredo\inst{1}\\ L.~Kowada\inst{2}\and F.~Marquezino\inst{1}\and D.~Posner\inst{1}\and R.~Portugal\inst{4}}
\authorrunning{A. Abreu et al.}
\institute{Universidade Federal do Rio de Janeiro -- \email{\{santiago, celina, franklin, posner\}@cos.ufrj.br} \and Universidade Federal Fluminense -- \email{\{lfignacio, luis\}@ic.uff.br} \and Universidade Federal do Esp\'{\i}rito Santo \and Laborat\'orio Nacional de Computa\c{c}\~ao Cient\'{\i}fica -- \email{\{tharsodf, portugal\}@lncc.br}}

\toctitle{Lecture Notes in Computer Science}
\tocauthor{Alexandre S. de Abreu, Lu\'is Felipe I. Cunha, Tharso D. Fernandes, Celina M.H. de Figueiredo, Luis Antonio B. Kowada, Franklin de L. Marquezino, Daniel F.D. Posner, Renato Portugal}
\maketitle

\begin{abstract}
A tessellation of a graph is a partition of its vertices into vertex disjoint cliques. A tessellation cover of a graph is a set of tessellations that covers all of its edges. The $t$-\textsc{tessellability} problem aims to decide whether there is a tessellation cover of the graph with $t$ tessellations.
This problem is motivated by its applications to quantum walk models, in especial, the evolution operator of the staggered model is obtained from a graph tessellation cover. 
We establish upper bounds on the tessellation cover number given by the minimum between the chromatic index of the graph and the chromatic number of its clique graph and we show graph classes for which these bounds are tight. 
We prove $\mathcal{NP}$-completeness for \textsc{$t$-tessellability} if the instance is restricted to planar graphs, chordal $(2,1)$-graphs, $(1,2)$-graphs, diamond-free graphs with diameter five, or for any fixed $t$ at least $3$. 
On the other hand, we improve the complexity for \textsc{$2$-tessellability} to a linear-time algorithm.

\end{abstract}

\section{Introduction}
Random walks play an important role in Computer Science mainly in the area of algorithms and it is expected that quantum walks, which is the quantum counterpart of random walks, will play at least a similar role in Quantum Computation. 
In fact, the interest in quantum walks has grown considerably in the last decades, especially because they can be used to build quantum algorithms that outperform their classical counterparts~\cite{Ven12}.

Recently, the staggered quantum walk model~\cite{PSFG16} was proposed. 
This model is defined by an evolution operator, which is described by a product of local unitary matrices obtained from a \emph{graph tessellation cover}. A tessellation is a partition of the vertices of a graph into vertex disjoint cliques, and a tessellation cover is a set of tessellations so that the union covers the edge set. In order to fully understand the possibilities of the staggered model, it is fundamental to introduce the \textsc{$t$-tessellability} problem. This problem aims to decide whether a given graph can be covered by $t$ tessellations.

The simplest evolution operators are the product of few local unitary matrices, and at least two matrices (corresponding to $2$-tessellable graphs) are required. 
There is a recipe to build a local unitary matrix based on a tessellation~\cite{PSFG16}. 
Each clique of the partition establishes a neighborhood around which the walker can move under the action of the local unitary matrix. 
To define the evolution operator of the quantum walk, one has to include extra tessellations until the tessellation union covers the edge set.
Fig.~\ref{fig:evolution} depicts an example of how a quantum walker could spread across the vertices of a graph, given a particular tessellation cover. Bold vertices represent non-zero amplitudes, meaning that a measurement of the position can reveal the walker at one of those vertices. Note that after each step the walker spreads across the cliques in the corresponding tessellation.

\begin{figure}
\centering
     \includegraphics[scale=0.3]{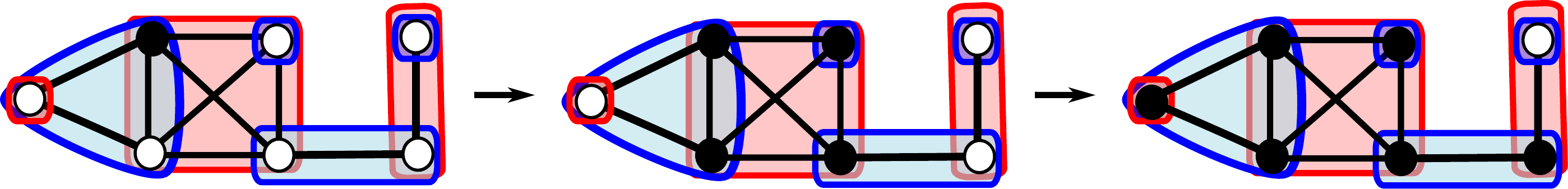}
     \caption{The spreading of a walker subject to locality across a $2$-tessellable graph. At each step, the walker may be observed at bold vertices.\label{fig:evolution}}
\end{figure}


The study of tessellations in the context of Quantum Computing was proposed by Portugal et al.~\cite{PSFG16} with the goal of obtaining the dynamics of quantum walks. Portugal analyzed the $2$-tessellable case in~\cite{Por16b} and we described examples for the $t$-tessellable case in~\cite{MatContemp}. The present work is the first systematic attempt to study the tessellation problem as a branch of Graph Theory.
Our aim is the study of graph classes with extremal 
tessellation cover numbers, efficient algorithms, and hardness.

We describe upper bounds in Section \ref{sec:two}. 
We establish graph classes for which these bounds are tight in Section~\ref{sect:exposed}. 
We use these graphs with extremal tessellation covers to establish hardness results for several graph classes in Section~\ref{sec:four}. 
We obtain proofs of \textsc{$t$-tessellability} $\mathcal{NP}$-completeness for planar graphs, chordal $(2,1)$-graphs, $(1,2)$-graphs, diamond-free graphs with diameter five, or if $t$ is fixed for $t \geq 3$. 
Moreover, we describe a linear-time algorithm for \textsc{$2$-tessellability}. 

\section{Preliminaries on the tessellation cover number}
\label{sec:two}
A \emph{clique} is a subset of vertices of a graph such that its induced subgraph is complete. 
The size of a maximum clique of a graph $G$ is denoted by $\omega(G)$. The \emph{clique graph} $K(G)$ is the intersection graph of the maximal cliques of~$G$. 
A \emph{partition of the vertices of a graph into cliques} is a collection of vertex disjoint cliques, where the union of these cliques is the vertex set. 
Clique graphs play a central role in tessellation covers. 
(See~\cite{Szw03} for an extensive survey on clique graphs and~\cite{BkWest} for omitted graph theory terminologies).

\begin{definition}
\emph{A \emph{tessellation} ${\mathcal{T}}$ is a partition of the vertices of a graph into cliques. 
An edge \emph{belongs} to the tessellation ${\mathcal{T}}$ if and only if its endpoints belong to the same clique in ${\mathcal{T}}$. 
The set of edges belonging to ${\mathcal{T}}$ is denoted by ${\mathcal{E}}({\mathcal{T}})$.}
\end{definition}

\begin{definition}
\emph{Given a graph $G$ with edge set $E(G)$, a \emph{tessellation cover} of size~$t$ of $G$ is a set of $t$~tessellations ${\mathcal{T}}_1,...,{\mathcal{T}}_t$, whose union $\cup_{i=1}^t\,{\mathcal{E}}({\mathcal{T}}_i)=E(G)$.
A graph~$G$ is called \emph{$t$-tessellable} if there is a tessellation cover of size at most~$t$.
The \textsc{$t$-tessellability problem} aims to decide whether a graph $G$ is $t$-tessellable.
The \emph{tessellation cover number} $T(G)$ is the size of a smallest tessellation cover of~$G$. }
\end{definition}

A \emph{coloring} (resp. an \emph{edge-coloring}) of a graph is a labeling of the vertices (resp. edges) with colors such that no two adjacent vertices (resp. incident edges) have the same color. 
A $k$-\emph{colorable} (resp. $k$\emph{-edge-colorable}) graph is the one which admits a \emph{coloring} (resp. an \emph{edge-coloring}) with at most $k$ colors. 
The \emph{chromatic number} $\chi(G)$ (resp. \emph{chromatic index} $\chi'(G)$) of a graph $G$ is the smallest number of colors needed to color the vertices (resp. edges) of $G$.
 
Note that an edge-coloring of a graph $G$ induces a tessellation cover of $G$. 
Each color class induces a partition of the vertex set into disjoint cliques of size two (vertices incident to edges of that color) and cliques of size one (vertices not incident to edges of that color), which forms a tessellation. 
Moreover, a coloring of $K(G)$ induces a tessellation cover of $G$. 
As presented in~\cite{MatContemp}, two vertices of the same color in $K(G)$ correspond to disjoint maximal cliques of $G$ and every edge of $G$ is in at least one maximal clique. 
So, each color in $K(G)$ defines a tessellation in $G$ by possibly adding cliques of size one (vertices that do not belong to maximal cliques of $G$, which are related to vertices of $K(G)$ with that color), such that the union of these tessellations is the edge set of $G$.
Hence, we have the following upper bounds.

\begin{theo} \label{prop:upper}
If $G$ is a graph, then $T(G) \le \min\left\{\chi'(G),\chi(K(G))\right\}$.
\end{theo}

Portugal~\cite{Por16b} characterized the $2$-tessellable graphs as those whose clique graphs are bipartite graphs. 
In order to characterize $t$-tessellable graphs, for $t\geq~3$, we were able to find graph classes such that $T(G) =3$, with $\chi'(G)$ and $\chi(K(G))$ arbitrarily large, and additionally graph classes with extremal values of Theorem~\ref{prop:upper}, i.e. $T(G) = \chi'(G)$ but $\chi(K(G))$ 
arbitrarily large; and $T(G) = \chi(K(G))$ but $\chi'(G)$ arbitrarily large, some of those examples were described in~\cite{MatContemp}.

An interesting case occurs for a triangle-free graph. 
Note that any of its tessellations can only be formed by cliques of size two or one.
Hence, we have an extremal result that if $G$ is a triangle-free graph, then $T(G) = \chi'(G)$.


A graph is $(k, \ell)$ if its vertex set can be partitioned into $k$ stable sets and $\ell$~cliques.
Particularly, $(2,0)$-graphs are \textit{bipartite graphs}. 
It is known how to $\Delta$-edge-color bipartite graphs~\cite{ArGarey} and to $\Delta$-edge-color $\{\textrm{triangle}, \textrm{proper major}\}$-free graphs~\cite{ArGuedes} in polynomial time. 
Therefore,  not only \textsc{$t$-tessellability} is polynomial time solvable for bipartite graphs and for $\{\textrm{triangle}, \textrm{proper major}\}$-free graphs, but there are also polynomial time procedures to obtain a minimum tessellation cover for these graph classes.
Besides that, it is known that \mbox{\textsc{$3$-edge colorability}} of triangle-free graphs is $\mathcal{NP}$-complete~\cite{Kor97}.
Therefore, \textsc{$3$-tessellability} of triangle-free graphs is also $\mathcal{NP}$-complete.


\section{Extremal tessellation covers}
\label{sect:exposed}


Throughout the paper, an extremal tessellation cover of a graph is one of its tessellation covers whose size reaches an upper bound of 
Theorem~\ref{prop:upper}. 
We are particularly interested in constructing graphs for which these extremal tessellation covers correspond (or are close) to an optimal edge coloring or an optimal vertex coloring of the clique graph of the original graph.
For the sake of convenience, we may omit one-vertex cliques inside tessellations in our proofs.


\begin{constr} \label{fact:addstar}
\emph{Let $H$ be obtained from a graph $G$ by adding a star with $\chi'(G)$ leaves and identifying one of these leaves with a minimum degree vertex of $G$.}
\end{constr}


The tessellation cover number of $H$, obtained from Construction~\ref{fact:addstar} on a non-regular graph $G$, is equal to its original chromatic index, i.e., $T(H) = \chi'(H) = \chi'(G)$.
For regular graphs, if $\chi'(G) = \Delta+1$, then $T(H) = \chi'(H)=\chi'(G)$. 
Otherwise, $T(H) = \chi'(H)=\chi'(G)+1$.
Construction~\ref{fact:addstar} also implies that every non-regular graph $G$ is subgraph of a graph~$H$ with $T(H) = \chi'(H)=\chi'(G)$.

Similarly, there exists a construction in diamond-free graphs, which are the $\{K_4 \backslash \{e\}\}$-free graphs, to force the tessellation cover number to be equal to the original chromatic number of the clique graph. 
First, we define a property of the cliques on a tessellation called exposed maximal clique. 
Such a property helps us with particular cases of diamond-free graphs.


\begin{definition}
\emph{A maximal clique $K$ of a graph $G$ is said \emph{exposed} by a tessellation cover ${\mathcal{C}}$ if $E(K) \not\subseteq \mathcal{E}(\mathcal{T})$ for all $\mathcal{T} \in \mathcal{C}$, that is, the edges of $K$ are covered by no tessellation of $\mathcal{C}$.}
\end{definition}


\begin{lemma}\label{prop:exposedk}
A graph $G$ admits a minimum tessellation cover with no exposed maximal cliques if and only if  $T(G) = \chi(K(G))$.
\end{lemma}

In the remaining part of this section we consider diamond-free graphs, which have the following properties~\cite{Pet03}: (1) their clique-graphs are diamond-free, and (2) any two maximal cliques intersect in at most one vertex.


\begin{theo}\label{prop:omega=chi}
If $G$ is a diamond-free graph with $\chi(K(G))=\omega(K(G))$, then $T(G)=\chi(K(G)).$
\end{theo}

A graph is \textit{$K$-perfect} if its clique graph is perfect~\cite{kperfect}. 
Since a diamond-free $K$-perfect graph $G$ satisfies the premises of Theorem~\ref{prop:omega=chi}, we have $T(G)=\chi(K(G))$.
Note that the size of the clique graph of a diamond-free graph is polynomially bounded by the size of the original graph. 
Moreover, there is a polynomial-time algorithm to obtain an optimal coloring of $K(G)$ with $\omega(K(G))$ colors~\cite{kperfectcol} and, by Theorem~\ref{prop:upper}, a coloring of $K(G)$ with $t$ colors yields that $G$ is $t$-tessellable. 
Thus, both the tessellation cover number and a minimum tessellation cover of diamond-free $K$-perfect graphs are obtained in polynomial time.

Interestingly, there are diamond-free graphs whose clique graphs have chromatic number greater than the tessellation cover number. 
Fig.~\ref{fig:myci} illustrates an example of a $3$-tessellable diamond-free graph whose clique graph has chromatic number~$4$ (the clique graph $K(G)$ is the Gr\"otzsch graph, i.e. Mycielskian of a $5$-cycle graph). 
Note that any minimum tessellation cover of this graph necessarily has an exposed maximal clique. 
Moreover, this graph shows that the upper bound of Theorem~\ref{prop:less4} is tight.

\begin{figure}
\centering
  \includegraphics[height=2.5cm]{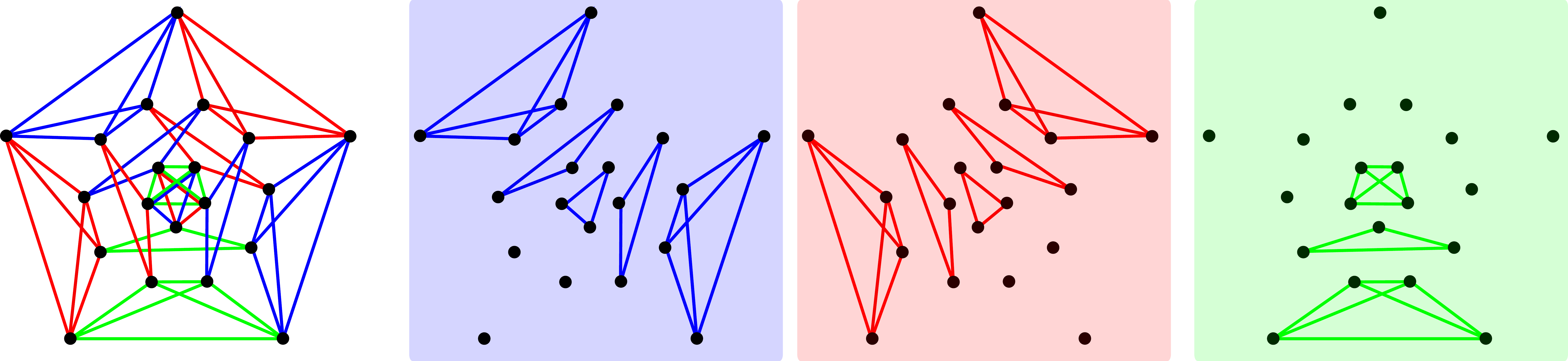}
  \caption{Example of a 3-tessellable graph $G$ whose clique graph is the Mycielskian of a $C_5$, with $\chi(K(G))=4$ but $T(G) = 3$. Each tessellation is depicted separately.\label{fig:myci}}
\end{figure}

\begin{lemma}\label{lemma:adjacentcliques}
Let $G$ be a $3$-tessellable diamond-free graph. 
If $C_1$ and $C_2$ are two maximal cliques of $G$ with a common vertex, then $C_1$ and $C_2$ cannot be both exposed by a minimum tessellation cover.
\end{lemma}

\begin{theo}\label{prop:less4}
If $G$ is a $3$-tessellable diamond-free graph, then $3 \le \chi(K(G))\le 4$.
\end{theo}

We finish this section with a construction which forces the tessellation cover number of a graph $H$, obtained from Construction~\ref{fact:addleaves} on a diamond-free graph $G$, to be $T(H) = \chi(K(H)) = \chi(K(G))$.
If $G$ has $T(G) < \chi(K(G))$, then there is no vertex that belongs to $\chi(K(G))$ maximal cliques. 
The graph $H$ obtained from $G$ by Construction~\ref{fact:addleaves} satisfies $\chi(K(H)) = \chi(K(G))$ and contains a vertex that belongs to $\chi(K(G))$ maximal cliques, which implies $T(H) = \chi(K(G))$. 

\begin{constr} \label{fact:addleaves}
\emph{Let $H$ be obtained from a graph $G$ by iteratively adding pendant vertices to a vertex of $G$ until it belongs to $\chi(K(G))$ maximal cliques.}
\end{constr}

Construction~\ref{fact:addleaves} implies that every diamond-free graph $G$ is a subgraph of a graph $H$ with $T(H) = \chi(K(H)) = \chi(K(G))$.
Note that this construction is not restricted to diamond-free graphs and it can also be applied several times to vertices that only belong to one maximal clique. The hardness proofs of Theorems~\ref{teo:np2} and \ref{prop:np21} rely on this result.


\section{Computational complexity}\label{sect:complexity}
\label{sec:four}

Now, we focus on the computational complexity of \textsc{$t$-tessellability}, by firstly proving that the problem is in $\mathcal{NP}$. In Section~\ref{sec:npc}, we use extremal tessellation covers obtained in the previous section to show $\mathcal{NP}$-completeness for any fixed $t \geq 3$ and when the problem instance is restricted to some graph classes. 
In Section~\ref{sec:poly}, we efficiently solve \textsc{$2$-tessellability} in linear time.

\begin{lemma}
\textsc{$t$-tessellability} is in $\mathcal{NP}$.
\end{lemma}
\begin{proof}
Consider a certificate for an instance of \textsc{$t$-tessellability}, which consists of at most $t$ tessellations of a given graph $G$.
Note that each tessellation has at most $m$ edges.
Moreover, by Theorem~\ref{prop:upper} and the well-known Vizing's theorem on edge-colorability, if $t\geq \Delta+1$, then the answer is automatically \emph{YES}.
Otherwise, one can easily verify in polynomial time if the at most $m$ edges in each of the at most $t \leq \Delta+1$ tessellations form disjoint cliques in $G$ and if the at most $m (\Delta+1)$ edges in these tessellations cover all edges of~$G$.~\qed
\end{proof}

\vspace{-0.5cm}
\subsection{$\mathcal{NP}$-completeness}
\label{sec:npc}

A graph is \textit{planar} if it can be embedded in the plane such that no two edges cross each other.
We show a polynomial transformation from the $\mathcal{NP}$-complete \textsc{$3$-colorability} of planar graphs with maximum degree four~\cite{ArGarey} to \textsc{$3$-tessellabi\-lity} of planar graphs with maximum degree six.

\begin{figure}
\centering
  \includegraphics[height=2.5cm]{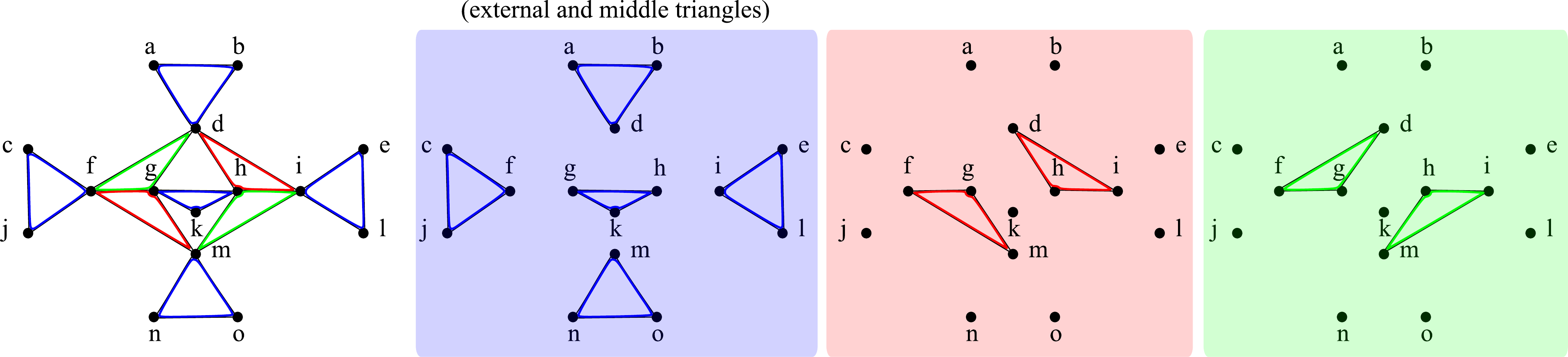}
  \caption{The 3-tessellable \textit{graph-gadget} of Lemma~\ref{lem:np1}.  Each tessellation is depicted separately. The external vertices are 
  $a,\ b, \ c, \ e, \ j, \ l, \ n, \ o$, and the internal vertices are the remaining ones.}
	\label{fig:Hajos_v1Mycielskib}
\end{figure}


\begin{lemma} 
Any tessellation cover of size 3 of the \emph{graph-gadget} depicted in Fig.~\ref{fig:Hajos_v1Mycielskib} contains a tessellation that has the middle and the external triangles.
\label{lem:np1} 
\end{lemma}


\begin{constr} \label{fact:planar}
\emph{Let graph $H$ be obtained from a graph $G$ by local replacements of the vertices of $G$ such that each vertex $u$ of $G$ represents a graph-gadget of Fig.~\ref{fig:Hajos_v1Mycielskib} denoted by $u$-gadget and each edge $uv$ of $G$ represents the intersection of the $u$-gadget with the $v$-gadget by identifying two external vertices of external triangles of those graph-gadgets. 
}
\end{constr}

\begin{theo}\label{prop:npplanar}
\textsc{$3$-tessellability} of planar graphs with $\Delta\leq 6$ is $\mathcal{NP}$-complete.
\end{theo}
\begin{proof}
Let $G$ be an instance graph of \textsc{$3$-colorability} of planar graphs with $\Delta \leq 4$ and $H$ be obtained by Construction~\ref{fact:planar} on $G$.
Notice that applying Construction~\ref{fact:planar} on a planar graphs with $\Delta\leq 4$ results on a planar graph with $\Delta\leq 6$.

Suppose that $G$ is $3$-colorable. Then, $H$ is $3$-tessellable because the middle and the external triangles of a $v$-gadget can be covered by the tessellation related to the color of $v$ and the remaining triangles of the $v$-gadget can be covered by the other two tessellations.


Suppose that $H$ is $3$-tessellable. Then, $G$ is $3$-colorable because the color of $v$ in $G$ can be related to the tessellation that covers the middle triangle of the $v$-gadget. This assignment is a 3-coloring because by Lemma~\ref{lem:np1} all external triangles of the $v$-gadget belong to the same tessellation of the middle triangle. The external triangles of the $v$-gadget are connected to the external triangles of the graph-gadgets of the neighborhood of $v$. Then, the tessellations of the latter external triangles must differ from the external triangles of the $v$-gadget.
This implies that the neighborhood of vertex $v$ receives different colors from the color of $v$.~\qed
\end{proof}


The next construction allows us to show a hardness proof of \textsc{$t$-tessellability} for any fixed $t \geq 4$.

\begin{constr} \label{fact:ttess}
\emph{Let $H$ be a graph obtained from a graph $G$ and a subset $F\subseteq V(G)$ as follows. Initially $H$ is equal to $G$. 
Let $F=\{v_1,...,v_{|F|}\}$ be a subset of vertices of $V(H)$.
Add to $H$ a complete graph $U=\{u_1,...,u_{|F|}\}$.
Add three vertices $c_1$, $c_2$ and $c_3$ adjacent to all vertices of $U$.
Consider an integer $t\ge 4$.
 For each $c_i$ $(1 \leq i \leq 3)$, add $t-1$ pendant vertices incident to $c_i$.
For each $1 \leq j \leq |F|$, add an edge $v_ju_j$ and vertices $w_{j,l}$ for $1 \leq l \leq t-3$ adjacent to both $v_j$ and $u_j$.
For each vertex $w_{j, l}$, add $t-1$ pendant vertices incident to $w_{j, l}$.}
\end{constr}

\begin{theo}\label{teo:np2}
\textsc{$t$-tessellability} for any fixed $t \geq 4$ is $\mathcal{NP}$-complete.
\end{theo}
\textit{Sketch of the proof.}
Let $G$ be an instance graph of \textsc{$3$-colorability} of planar graphs with $\Delta \leq 4$. Let $H'$ be the graph obtained from Construction~\ref{fact:planar} on~$G$. Let $H$ be the graph obtained from Construction~\ref{fact:ttess} on~$H'$ with $F$ being the set of all internal vertices of all graph-gadgets of $H'$.

Let $H[H']$ be the induced subgraph of $H$ by the vertices of $H'$.
$H$ is $t$-tessellable (for $t \geq 4$) if and only $H[H']$ is $3$-tessellable and the ${\mathcal{NP}}$-completeness follows imediately from Theorem~\ref{prop:npplanar}.~\qed

Next, we show a polynomial transformation from the $\mathcal{NP}$-complete \textsc{$3$-colo\-ra\-bility}~\cite{ArGarey} to \textsc{$4$-tessellability} of chordal $(2,1)$-graphs.
This proof is based on a result of Bodlaender et al.~\cite{ArBod} for \textsc{$3$-$L(0,1)$-coloring} of split graphs.

\begin{constr} \label{fact:npc21cordal}
\emph{Let $H$ be a graph obtained from a non-bipartite graph $G$ as follows. Initially $V(H)=V(G)\cup E(G)$ and $E(H)=\emptyset$.
Add edges to $H$ so that the $E(G)$ vertices induce a clique. For each $e=vw \in E(G)$, add to $H$ edges $ve$ and $we$.
For each vertex $v\in V(H)\cap V(G)$, add three pendant vertices incident to $v$. 
Add a vertex $u$ adjacent to all $E(G)$ vertices. Add three pendant vertices incident to $u$.  Denote all pendant vertices by $V_2$. }
\end{constr}


\begin{theo}
\label{prop:np21}
\textsc{$4$-tessellability} of chordal $(2,1)$-graphs is $\mathcal{NP}$-complete.
\end{theo}
\textit{Sketch of the proof.}
Consider the graph $H$ obtained by Construction~\ref{fact:npc21cordal} on a non-bipartite  instance graph $G$ of \textsc{$3$-colorability}.
We have $V(H) = V(G) \cup E(G) \cup V_2 \cup \{u\}$.
Clearly, $H$ is chordal and $(2,1)$ with $E(G)$ as a clique, $V(G) \cup \{u\}$ as a stable set, and $V_2$ as another stable set.

The key idea of the proof is that the pendant vertices $V_2$ force the maximal cliques incident to the vertices in $V(G)\cup \{u\}$ to be non-exposed.
Notice that the tessellation used to cover the clique $E(G)\cup \{u\}$ cannot cover any other maximal clique with size greater than one incident to $V(G)$ vertices and $E(G)$ vertices in $H$. 
Therefore, there are only three remaining tessellations to cover those maximal cliques.
This implies that $G$ is $3$-colorable if and only if $H$ is $4$-tessellable because if $uv \in E(G)$, then they receive different colors in a $3$-coloring of~$G$. The maximal clique containing vertices $E(G)\cup \{u\}$ and the maximal clique containing vertices $E(G)\cup \{v\}$ share a same neighborhood $uv$ and must be covered by different tessellations.~\qed

\begin{constr} \label{fact:npc12}
\emph{Let $H'$ be a graph obtained from the graph of Construction~\ref{fact:npc21cordal} by transforming the stable set of $V(G)$ into a clique, removing one pendant vertex of each vertex of $V(G)$, and adding a vertex $u'$ adjacent to all vertices of $V(G)$ with three new pendant vertices incident to it.}
\end{constr}

Clearly, $H'$ from Construction~\ref{fact:npc12} is a $(1,2)$-graph.
Observe that $H'$ is $4$-tessellable if and only if $H$ (from Theorem~\ref{prop:np21}) is $4$-tessellable.
Therefore, \textsc{$4$-tessellability} is $\mathcal{NP}$-complete for $(1,2)$-graphs.

Next, we show a polynomial transformation from the $\mathcal{NP}$-complete problem \textsc{NAE $3$-SAT}~\cite{ArGarey} to \textsc{$3$-tessellability} of diamond-free graphs with diameter five. 
This proof is given in two phases: given an instance~$I$ of \textsc{NAE $3$-SAT} we construct a clique graph $K(G)$ for which we show that there is a $3$-coloring of $K(G)$ if and only if~$I$ is satisfiable; subsequently, we show that there is a construction of a graph with diameter five $G$ for which $G$ is $3$-tessellable if and only if $K(G)$ is $3$-colorable.

\begin{constr} \label{fact:nae1}
\emph{Let $K(G)$ be a graph obtained from an instance of \textsc{NAE $3$-SAT} as follows.
For each variable $v$ of $I$, include a $P_2$ with vertices $v$ and $\overline{v}$ in $K(G)$.
Moreover, add a vertex $u$ adjacent to all $P_2$'s vertices.
And, for each clause $\{a \lor b \lor c\}$ of $I$, add a triangle with vertices $T_a, T_b, T_c$ in $K(G)$ and three edges $aT_a$, $bT_b$, and $cT_c$.}
\end{constr}

\begin{lemma} Let $K(G)$ be obtained from Construction~\ref{fact:nae1} on a \textsc{NAE $3$-SAT} instance $I$. Then $K(G)$ is $3$-colorable if and only if~$I$ is satisfiable.
\label{lem:npb1} 
\end{lemma}
\begin{proof}
Note that, w.l.o.g., the color $1$ given to the vertex $u$ in a  $3$-coloring cannot be used in any vertex of a $P_2$. 
Moreover, each of the literal vertices $v$ and $\overline{v}$ of a $P_2$ receives either the color $2$ or $3$.
Assume w.l.o.g that a literal is true if its color is~$2$, and false otherwise.

If $K(G)$ is $3$-colorable, then there are no three vertices connected to a clause's triangle with the same color.
Otherwise, this color would not be used in the triangle vertices and $K(G)$ would not be $3$-colorable, a contradiction.
Therefore, the above assignment of values to literals would give a satisfiable solution to the instance.

Conversely, if $I$ is satisfiable, then one may assign color $2$ to each literal vertex which is true and color $3$ to its negation.
Moreover, vertex $u$ receives color~$1$.
Since there are no three literal vertices with the same color adjacent to the clause triangles, one may assign colors to the vertices of the triangles in a $3$-coloring where a vertex of the triangle adjacent to a vertex with color $2$ receives color $3$, and a vertex adjacent to a vertex with color $3$ receives color $2$. 
The other vertex receives color $1$.
\qed
\end{proof}


Next, we construct a graph $G$ for which its clique graph is the $K(G)$ obtained from Construction~\ref{fact:nae1}.


\begin{constr} \label{fact:nae2}
\emph{Let $G$ be obtained from its clique graph $K(G)$ (of Construction~\ref{fact:nae1}) as follows.
For each clause's triangle in $K(G)$, add a star with three leaves in $G$, where each of those leaves represents a literal of this clause.
Next, all $P_2$'s triangles in $K(G)$ are represented in $G$ by a clique $C$ of size the number of $P_2$'s.
Each vertex of this clique $C$ represents a variable of $K(G)$.
For each vertex $v$ of $C$ include the edges of two other cliques (one for each literal of the variable $v$) composed by the leaves of the stars which represent the literals $v$ and $\overline{v}$ and the vertex $v$ of $C$, as depicted in Fig.~\ref{fig:construction8}.}
\end{constr}

\begin{figure}[!h]
\begin{center}
  \includegraphics[width=11cm]{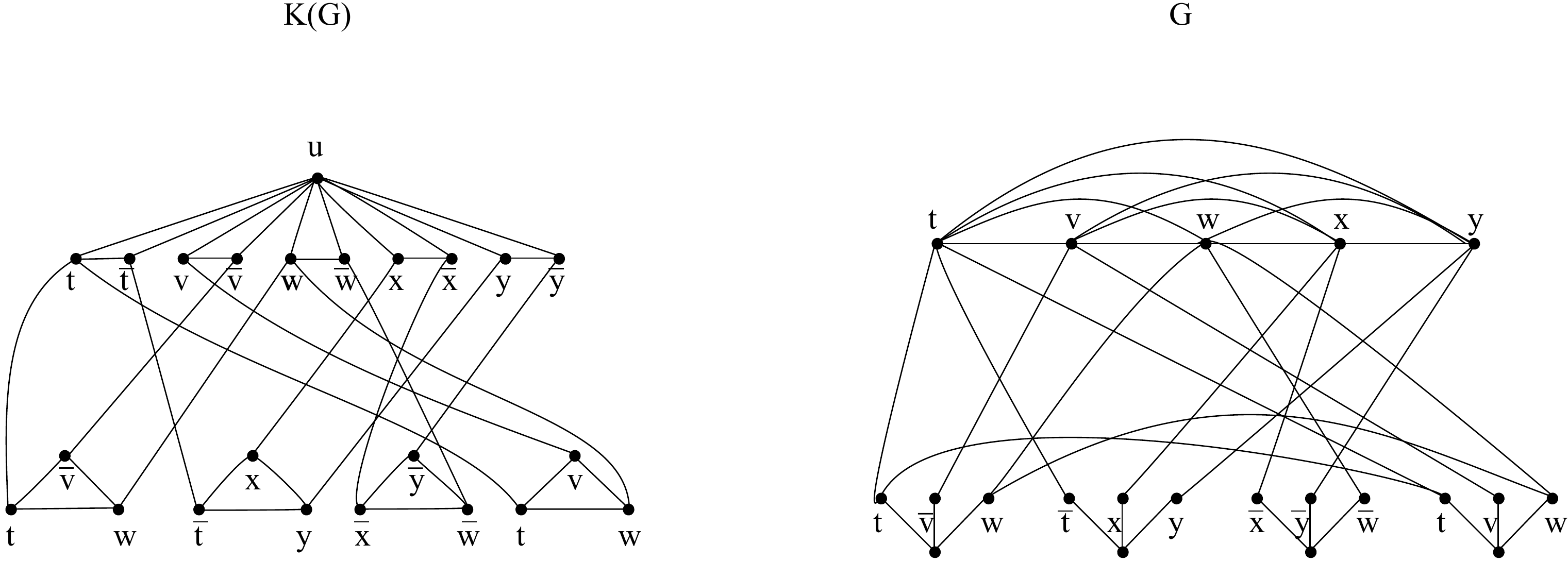}
  \caption{Example of Construction~\ref{fact:nae2}.}
  \label{fig:construction8}
\end{center}
\end{figure}

\begin{lemma} Let $K(G)$ be obtained by Construction~\ref{fact:nae1} on a \textsc{NAE $3$-SAT} instance~$I$ and $G$ be obtained by Construction~\ref{fact:nae2} on $K(G)$. Then $G$ is $3$-tessellable if and only if~$K(G)$ is $3$-colorable.
\label{lem:npb2} 
\end{lemma}
\begin{proof}
If $G$ is $3$-tessellable, we need one tessellation to cover the maximum clique whose size is the number of variables.
Therefore, the other two tessellations are used by the other two maximal cliques (which represent the literals of each variable). 
Moreover, the star of three leaves of each clause also needs to be covered by $3$ tessellations.
Note that these maximal cliques represent vertices in $K(G)$ and the tessellations represent their colors.
Therefore, $K(G)$ is $3$-colorable.

If $K(G)$ is $3$-colorable, then use these three colors as a guide to obtain a $3$-tessellation of $G$, where each color class of $K(G)$ represents that these maximal cliques of $G$ are covered by the tessellation which represents this color.~\qed
\end{proof}

Clearly, the graph $G$ obtained from Construction~\ref{fact:nae2} is diamond-free with diameter five.
Therefore, by Lemmas~\ref{lem:npb1} and ~\ref{lem:npb2}, the next theorem follows.

\begin{theo}\label{prop:nae3sat}
\textsc{$3$-tessellability} of diamond-free graphs with diameter five is $\mathcal{NP}$-complete.
\end{theo}



\subsection{{\sc$2$-tessellability}}\label{sec:poly}

Portugal~\cite{Por16b} showed that a graph $G$ is $2$-tessellable if and only if $K(G)$ is a bipartite graph.
Moreover, Peterson~\cite{Pet03} showed that $K(G)$ is bipartite if and only if $G$ is the line graph of a bipartite multigraph. 
Hence, determine if $G$ is $2$-tessellable is equivalent to verifying if $G$ is the line graph of a bipartite multigraph.

Protti and Szwarcfiter~\cite{ProSzw02} showed an $O(n^2m)$ time algorithm to decide if the clique graph of a given graph is bipartite. 
Moreover, Peterson~\cite{Pet03} showed an $O(n^3)$ time algorithm to decide if $G$ is the line graph of a bipartite multigraph.

The key idea of Peterson's algorithm is to group true twin vertices of a same clique of $G$, which represent multiedges in the bipartite multigraph $H$, where $G = L(H)$.
Then, it removes all those true twin vertices in each group but one, and the resulting graph is a line graph of a bipartite simple graph if and only if $K(G)$ is a bipartite graph.
To verify if a graph is a line graph of a bipartite graph, the Roussopoulos' linear-time algorithm is used~\cite{ArRou}. 

We improve Peterson's algorithm~\cite{Pet03}, by showing a faster way to remove true twin vertices 
belonging to a clique of a graph using its modular decomposition.
Throughout this section, we use notations of modules of a graph given in~\cite{ArSpin}.

Let $\mathcal{F}$ be the family of bipartite multigraphs obtained by adding multiple edges to $C_4$, $S_n$ or $P_4$. 
In order to make a modular decomposition of a graph $G$, we only consider graphs $G$ which are not a line graph of a graph in $\mathcal{F}$. If $G$ is a line graph of a graph in $\mathcal{F}$, we can consider this case separately, and easily achieve linear time.
Note that there are bipartite multigraphs with a same line graph. 
Therefore, we only consider the ones which maximize the number of multiple edges. 
Moreover, we only consider connected graphs, since the tessellation cover number of a disconnected graph is the maximum among the parameter on its connected components.

\begin{lemma} Let $H$ be a bipartite multigraph not in $\mathcal{F}$ and $L(H)$ be its line graph. 
Two edges $e_1$ and $e_2$ with same extremes in $H$ represent vertices in a same maximal strong module of $L(H)$ with size less than $|V(L(H))|$.
\label{lem:multbip1} 
\end{lemma}

\begin{lemma} Let $H$ be a bipartite multigraph not in $\mathcal{F}$ and $L(H)$ be its line graph. 
Any maximal strong module in a modular decomposition of $L(H)$ with size less than $|V(L(H))|$ induces a clique in $L(H)$.
\label{lem:multbip2} 
\end{lemma}

\begin{theo} \label{teo:2tess}
\textsc{$2$-tessellability} can be solved in linear time.
\end{theo}
\begin{proof}
First, we use McConnell and Spinrad's linear-time algorithm to obtain a modular decomposition of $G$.
By Lemmas~\ref{lem:multbip1}~and~\ref{lem:multbip2}, we know that the strong modules in any modular decomposition of a line graph of a bipartite multigraph $H\notin\mathcal{F}$ induce cliques. 
Moreover, the vertices of these cliques in $L(H)$ are related to edges of $H$ with same extremes.

Then, we check if each of at most $O(|V(G)|)$ strong modules induces cliques in $G$, which can be done in $O(|V(G)|+|E(G)|)$. Otherwise, we know that $G$ is not a line graph of a bipartite multigraph.
Next, we remove all true twins vertices in each strong modules but one, obtaining the graph $G'$.
This step is related to remove all multiedges of $H$ which share same extremes.
Therefore, the graph $G$ is a line graph of a bipartite multigraph $H$ if the resulting graph $G'$ is a line graph of a bipartite simple graph $H'$.

Finally, we use Roussopoulos' linear-time algorithm to determine if $G'$ is a line graph, and if so, obtain its root graph $H'$ for which $G'$ is the line graph. 
Note that verifying to if $H'$ is a bipartite graph can be done in linear time by using a breadth-first search (because the size of the root graph of $G'$ is asymptotically bounded by the size of~$G'$).
\qed
\end{proof}

\section{Concluding remarks and discussion}
\label{sec:conc}


We investigate extremal results on graph tessellation covers, which are fundamental for the development of quantum walks in the staggered model. These results help to understand the complexity of the unitary operators necessary to express the evolution of staggered quantum~walks.
We establish tight upper bounds for the tessellation cover number of a graph $G$ related to $\chi'(G)$ and $\chi(K(G))$ and 
we determine graph classes which reach these extremal bounds.
This study provides tools to distinguish several classes for which the \textsc{$t$-tessellability} problem is efficiently tractable (bipartite, diamond-free $K$-perfect and $\{$triangle, proper major$\}$-free graphs) from others the problem is $\mathcal{NP}$-complete (planar, triangle-free, chordal $(2,1)$-graphs, $(1,2)$-graphs, and diamond-free graphs with diameter five). 
We also establish the \textsc{$t$-tessellability} $\mathcal{NP}$-completeness for any fixed $t\geq 3$. 
Moreover, we improve to linear time the known algorithms to recognize $2$-tessellable graphs~\cite{Por16b}, line graphs of bipartite multigraphs~\cite{ProSzw02}, and graphs~$G$ such that $K(G)$ is bipartite~\cite{Pet03}.
As a consequence, we establish an interesting complexity dichotomy between {\sc edge-colorability} and $t$-{\sc tessellability}: {\sc edge-colorability} of planar graphs with $\Delta\geq 8$ is in $\mathcal{P}$~\cite{sanders2001planar}, while $t$-{\sc tessellability} is $\mathcal{NP}$-complete 
(Theorem~\ref{prop:npplanar} replacing each of the four non external triangles that shares two vertices of external triangles by $K_4$'s) 
and; {\sc edge-colorability} of line graph of bipartite graphs is $\mathcal{NP}$-complete~\cite{cai1991np}, while $t$-{\sc tessellability} is in $\mathcal{P}$ (Theorem~\ref{teo:2tess}). We have not managed yet to establish the same dichotomy to $k$-{\sc colorability of clique graph} and $t$-{\sc tessellability}. 
We are currently trying to establish the hardness of the problem for $(1,1)$-graphs (split graphs). However, the computational complexity for $(0, 2)$-graphs (complement of bipartite graph) is still widely open.

A question that naturally arises is whether every graph has a minimum tessellation cover such that every tessellation contains a maximal clique. 
Although we believe in most cases the answer is true, we have computationally found a surprising example of a graph, which is depicted in Fig.~\ref{fig:maximal}, with all minimum tessellation covers requiring a tessellation without maximal cliques. 
We are currently trying to establish an infinite family of graphs for which this property does not hold and to establish other graph classes where it holds.%
The computational verification was performed through a reduction from \textsc{$t$-tessellability} problem to \textsc{set-covering} problem (the description of \textsc{set-covering} is available at~\cite{ArGarey}), where the finite set is the edge set of the input graph, and the
family of subsets consists of the edge subsets corresponding to all possible tessellations of the input graph. 
Another interesting issue is that two minimum tessellation covers may present different quantum walk dynamics. 
Therefore, we intend to study the different tessellation covers using the same number of tessellations, which may result in simpler quantum walks and more efficient quantum algorithms. More recently, a general partition-based framework for quantum walks has been proposed~\cite{konno2017partition}.

\begin{figure}[!h]
\vspace{-0.5cm}
\centering
  \includegraphics[width=12.5cm]{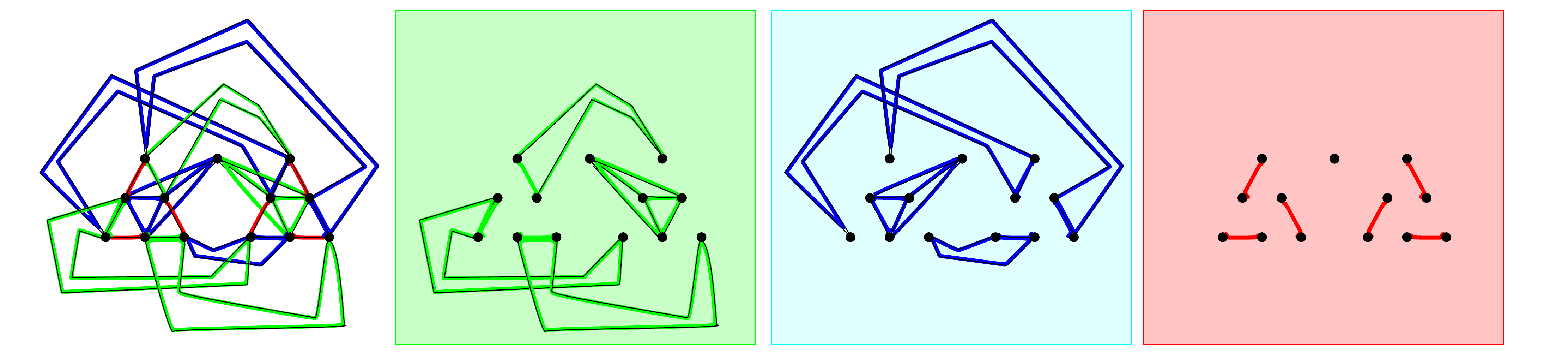}
  \caption{$3$-tessellable graph. Rightmost tessellation does not contain a maximal clique.}
  \label{fig:maximal}
\end{figure}



\begin{thebibliography}{10}
\expandafter\ifx\csname url\endcsname\relax
  \def\url#1{\texttt{#1}}\fi
\expandafter\ifx\csname urlprefix\endcsname\relax\def\urlprefix{URL }\fi
\expandafter\ifx\csname href\endcsname\relax
  \def\href#1#2{#2} \def\path#1{#1}\fi

\bibitem{Ven12}
S.~Venegas-Andraca, Quantum walks: a comprehensive review, Quantum Inf. Process
  11~(5) (2012) 1015--1106.

\bibitem{PSFG16}
R.~Portugal, R.~A.~M. Santos, T.~D. Fernandes, D.~N. Gon{\c{c}}alves, The
  staggered quantum walk model, Quantum Inf. Process 15~(1) (2016) 85--101.

\bibitem{Por16b}
R.~Portugal, Staggered quantum walks on graphs, Phys. Rev. A 93 (2016) 062335.

\bibitem{MatContemp}
A.~Abreu, L.~Cunha, T.~Fernandes, C.~de~Figueiredo, L.~Kowada, F.~Marquezino,
  D.~Posner, R.~Portugal, Bounds and complexity for the tessellation problem,
  Mat. Contemp. (accepted), (2017).

\bibitem{Szw03}
J.~L. Szwarcfiter, A survey on clique graphs, in: Recent Advances in Algorithms
  and Combinatorics, Springer, New York, 2003, pp. 109--136.

\bibitem{BkWest}
D.~West, Introduction to Graph Theory, Pearson, 2000.

\bibitem{ArGarey}
M.~R. Garey, D.~S. Johnson, Computers and Intractability: A Guide to the Theory
  of NP-Completeness, W. H. Freeman Co, 1979.

\bibitem{ArGuedes}
L.~M. Zatesko, R.~Carmo, A.~L.~P. Guedes, Edge-colouring of triangle-free
  graphs with no proper majors, II Encontro de Teoria da Computa\c{c}\~ao
  (2017) 71--74.

\bibitem{Kor97}
D.~P. Koreas, {The NP-completeness of chromatic index in triangle free graphs
  with maximum vertex of degree 3}, Appl. Math. Comput. 83~(1) (1997) 13--17.

\bibitem{Pet03}
D.~Peterson, Gridline graphs: a review in two dimensions and an extension to
  higher dimensions, Discrete Appl. Math. 126~(2--3) (2003) 223--239.

\bibitem{kperfect}
F.~Bonomo, G.~Dur\'an, M.~Groshaus, J.~Szwarcfiter, On clique-perfect and
  \uppercase{K}-perfect graphs, Ars. Combinatoria 80 (2006) 97--112.

\bibitem{kperfectcol}
M.~Gr\"otschel, L.~Lov\'asz, A.~Schrijver, Geometric Algorithms and
  Combinatorial Optimization, Springer-Verlag, 1988.

\bibitem{ArBod}
H.~Bodlaender, T.~Kloks, B.~Richard, J.~van Leeuwen, Approximation for
  lambda-colorings of graphs, Comput. J. 47 (2004) 1--12.

\bibitem{ProSzw02}
F.~Protti, J.~L. Szwarcfiter, Clique-inverse graphs of bipartite graphs, J.
  Combin. Math. Combin. Comput. 40 (2002) 193--203.

\bibitem{ArRou}
N.~D. Roussopoulos, {A max \{m,n\} algorithm for determining the graph H from
  its line graph G}, Inform. Process. Lett. 2 (1973) 108--112.

\bibitem{ArSpin}
R.~McConnell, J.~Spinrad, Linear-time modular decomposition and efficient
  transitive orientation of comparability graphs, Proceedings of the Fifth
  Annual ACM-SIAM Symposium on Discrete Algorithms 2 (1994) 536--545.

\bibitem{sanders2001planar}
D.~P. Sanders, Y.~Zhao, Planar graphs of maximum degree seven are class
  \uppercase{I}, J. Comb. Theory B 83~(2) (2001) 201--212.

\bibitem{cai1991np}
L.~Cai, J.~A. Ellis, \uppercase{NP}-completeness of edge-colouring some
  restricted graphs, Discrete Appl. Math. 30~(1) (1991) 15--27.

\bibitem{konno2017partition}
N.~Konno, R.~Portugal, I.~Sato, E.~Segawa, Partition-based discrete-time
  quantum walks, arXiv:1707.07127.

\end{thebibliography}

\end{document}